\documentclass[12pt]{article}
\usepackage{graphicx}
\usepackage{bbm}
\usepackage{amsmath}
\numberwithin{equation}{section}
\usepackage{natbib}
\usepackage{theorem} 
\usepackage{amsfonts}
\usepackage{bm}
\usepackage{amssymb}
\usepackage{accents}
\usepackage{hyperref}
\usepackage{xcolor}
\usepackage[onehalfspacing]{setspace}
\usepackage{xcolor}

\usepackage{tikz}
\usetikzlibrary{trees}
\usetikzlibrary{intersections}
\usetikzlibrary{automata}
\usetikzlibrary{patterns}
\usepackage{fullpage}

\newtheorem{theorem}{Theorem}[section]

\newtheorem{assumption}[theorem]{Assumption}

\newtheorem{claim}[theorem]{Claim}

\newtheorem{corollary}[theorem]{Corollary}

\newtheorem{lemma}[theorem]{Lemma}

\newtheorem{proposition}[theorem]{Proposition}

\newenvironment{proof}[1][Proof]{\noindent\textbf{#1.} }{\ \rule{0.5em}{0.5em}}
\begin{document}
	
	\title{Long run consequence of p-hacking
	}
	
	\author{Xuanye Wang\thanks{Institute for Advanced Economics Research, Dongbei University of Finance and Economics. \href{xuanye.wang@dufe.edu.cn}{xuanye.wang@dufe.edu.cn}.}}
	\date{\today}
	\maketitle

\begin{abstract}
	We study the theoretical consequence of p-hacking on the accumulation of knowledge under the framework of mis-specified Bayesian learning. 
	A sequence of researchers, in turn, choose projects that generate noisy information in a field. In choosing projects, researchers need to carefully balance as projects generates big information are less likely to succeed. In doing the project, a researcher p-hacks at intensity $\varepsilon$ so that the success probability of a chosen project increases (unduly) by a constant $\varepsilon$. In interpreting previous results, researcher behaves as if there is no p-hacking because the intensity $\varepsilon$ is unknown and presumably small. 
	We show that over-incentivizing information provision leads to the failure of learning as long as $\varepsilon\neq 0$. If the incentives of information provision is properly provided, learning is correct almost surely as long as $\varepsilon$ is small.
\end{abstract}

\section{Introduction}
P-hacking has attracted the attention of the entire academia.
There are tons of papers documenting that p-hacking is widespread throughout science and discussing how to catch and alleviate such behavior. To name a few, see \cite{GM2018}, \cite{BCFL2023} and \cite{HHLKJ2015} for example.
However, there is almost no paper studying the theoretical consequence of p-hacking on the accumulation of knowledge
\footnote{The only exception we know is \cite{HHLKJ2015}. In this empirical paper, they conclude that the extent of p-hacking is weak relative to the real effect-size in meta-analysis and suggest that p-hacking may not significantly alter the scientific consensus drawn from meta-analysis.}
. A paper p-hacks doesn't mean the theory of that paper is wrong. In many fields, knowledge is obtained after carefully weighing all the positive and negative evidences provided by many papers. If both sides p-hack at roughly the same intensity, will the effect of p-hacking cancel out? In occasional cases, p-hacking might even be a good thing, as it helps to bring more attention to a stunted but correct theory. A famous example is that Gregor Mendel might p-hack in his peas experiments.

We propose to study the consequence of p-hacking under the framework of mis-specified Bayesian learning. (See \cite{Bohren2016}, \cite{BH2021} and \cite{FII2020}) This literature studies whether Bayesian players eventually correctly learn when their specification of the underlying data-generating distribution is wrong. It exactly deals the problem of learning from literature suffered from p-hacking. Papers could be viewed as evidences sampled from an underlying distribution, and p-hacking distorted the arrival probability of each sampled evidence.

We propose the following simple model, hoping that it could shed some light in this direction. In a research area, one of two theories ($A, B$) is true. A sequence of researchers, in turn, choose one project to work on. Each project is indexed by a positive number $l$, meaning that the success of this project would generate likelihood ratio $l$ ($B$ over $A$) upon success. If the chosen project succeeds, the researcher is rewarded by $P(I)$. Here $I$ is the information generated (measured as the KL-divergence of beliefs before and after the project) and $P$ is an increasing payoff function. The true success probability of project $l$ is determined by state-contingent bell-shaped functions so that projects generating big information are less likely to succeed. Therefore, in choosing a project, a researcher needs to carefully weigh between the information generated and the success probability. We assume all researchers just maximize the expected payoff. As the project going, maybe because the p-value is slightly above $5\%$, the researcher painfully decides to p-hack. We model p-hacking as it unduly increase the success probability by a small constant $\varepsilon$. 
Such behavior makes the presented evidence stronger than it actually is. However, as $\varepsilon$ is unknown and presumably small, people still read the evidence as if there is no p-hacking.

Our model turns into a standard Bayesian learning model if the p-hacking intensity $\varepsilon=0$. In this case, after observing enough evidences, researchers would assign all the weight to the true theory. It would be great if a small change in $\varepsilon$ doesn't affect this. We find that it depends on the growth rate of the payoff function: if the payoff function grows fast enough, then the event that all the weight eventually gets assigned to the true theory happens with probability $0$ for all $\varepsilon>0$; if the payoff function grows slowly, the good result under $\varepsilon=0$ is restored for sufficiently small $\varepsilon$.

Let us briefly explain the intuition behind these results. If enough weight has been assigned to one theory, keep pushing beliefs toward this theory generates almost no information. So researchers are naturally motivated to go contrarian to generate a big information. The only thing preventing them from doing so is that the success probability of an extreme contrarian project is super small. But this can be compensated if the payoff function grows fast enough. In other words, if the current belief assigns heavy weight to theory $A$, a fast-growing payoff function provides enough incentives for researcher to work on a project strongly favoring theory $B$. Mathematically, we could keep track the belief at period $t$ as the log likelihood ratio of $B$ over $A$ conditional on what happened up to period $t$. Under a fast-growing payoff function, the expected change of this belief process turns positive once enough weight is on $A$. This prevent the belief process from going to $-\infty$ (assigning all the weight to $A$) even if theory $A$ is true. On the other hand, if the payoff function grows slowly, researchers would prefer to work on projects with large success probabilities, when enough weight is assigned to one theory. If theory $A$ is true and $\varepsilon$ is small, the above-mentioned belief process can be shown to be a supermartingale. We can then prove this supermartingale does go to $-\infty$ a.s. using martingale central limit theorem.

The primary contribution of this paper is suggesting using the mis-specified Bayesian learning framework to study the theoretical consequence of p-hacking on the accumulation of knowledge. To the best of our knowledge, no previous research goes in this direction. The secondary contribution is that we show over-incentivizing information provision could harm the learning in the presence of p-hacking. This result is also one of the few results in mis-specified learning literature that a $\epsilon$-small mis-specification leads to the failure of learning. We also show that p-hacking at low intensity doesn't harm the learning under proper incentives.

This paper is organized as following: in section $2$ we describe the model; in section $3$ we solve the model; in section $4$ we discuss how the growth rate of payoff functions affects researchers' choices of projects; in section $5$ we discuss potential extensions of the model and conclude.

\section{Model}
\label{section1.1-2023-10-22}
There is a research area with unknown states (theories) $A$ and $B$ and prior belief $(u_0,1-u_0)$ (here $u_0=\Pr(A)$). Without loss of generality, assume the true state is $A$.
In each period, one researcher arrives with probability $p<1$.
The researcher needs to choose a project $l$. The project, if succeed, will generate a likelihood ratio $l$ (state B over state A).
The probability that project $l$ succeeds is determined by a commonly known state-dependent function ($p_A(l)$ under state $A$, $p_B(l)$ under state $B$).

Upon the success of project $l$, the researcher is rewarded according to the amount of information generated, measured by the KL-divergence of beliefs before and after his project.
To be specific, information generated by project $l$ under belief $u$ is computed as 
\begin{eqnarray}
	&&I(u,l)=KL(u_{t+1}(l,u_t)|u_t=u)\notag\\
	&=&
	\frac{u}{u+(1-u)l}\log\frac{1}{u+(1-u)l}+\frac{(1-u)l}{u+(1-u)l}\log\frac{l}{u+(1-u)l}\notag
\end{eqnarray}
The researcher's payment is $P(I(u,l))$ with $P(I)$ being a strictly increasing function.
If the chosen project fails, the researcher receives nothing despite that the failure also provides some information.

In period $t$, the researcher, who knows current $u$, chooses project $l$ to maximize his expected payoff
\begin{eqnarray}
	\label{eqn1.1}
	\max_{l\in (0,+\infty)} P(I(u,l))[up_A(l)+(1-u)p_B(l)].
\end{eqnarray}
Whether the chosen project succeed is publicly observed by all researchers arrive in later periods. Besides, we assume there is a commonly known tie-breaking rule. Thus, even if the chosen project fails, all the future researchers could still correctly infer researcher $t$'s choice.

Up to this point, p-hacking hasn't shown up. 
In choosing the project, the researcher acts genuinely, without considering the possibility that he may p-hack in the future. However, as the project goes, maybe his p-value is just slightly above $5\%$ and he makes a painful decision to p-hack a little bit.

We model p-hacking as it slightly increases the succeed probability of a chosen project by a constant $\varepsilon$. However, as this $\varepsilon$ is unknown and presumably small, people treat past results as if there wasn't any p-hacking.
In other words, p-hacking distorts the learning process in the following way: if project $l^*$ succeeds, which happens with probability $p(p_A(l^*)+\varepsilon)$, the likelihood ratio of state $B$ over $A$ is updated from $\frac{1-u}{u}$ to $\frac{1-u}{u}l^*$; if project $l^*$ fails, which happens with probability $1-p(p_A(l^*)+\varepsilon)$, the likelihood ratio is updated to $\frac{1-u}{u}\frac{1-pp_B(l^*)}{1-pp_A(l^*)}$.

We care about whether the beliefs, updated in the distorted way with p-hacking intensity $\varepsilon$, eventually assigns all the weight to the true state $A$. That is, whether the stochastic belief process $\frac{1-u_t^\varepsilon}{u_t^\varepsilon}$ converges to $0$. (Or equivalently, $u_t^\varepsilon$ converges to $0$.)

Some extra assumptions are made about functions of success probabilities $p_A(l),p_B(l)$ and function of payment $P(I)$. First, we assume that the success probability gradually dies out as the likelihood ratio $l$ gets extreme. That is, it is harder to get more convincing evidences, and impossible to get fully revealing evidences. Mathematically, the assumption is written as
\begin{assumption}
	\label{assumption1.1}
	\begin{enumerate}
		\item $\exists l_A$ near $1$ such that
		$$p_A'(l)>0 \mbox{ on } (0,l_A); p_A'(l)<0 \mbox{ on } (l_A,+\infty);$$
		similarly, exists $l_B$ near $1$ with similar property.
		\footnote{It would be desirable to assume that $p_A(l)$ and $p_B(l)$ both peak at $1$. That is, the project that provides no information succeeds with the highest probability. However, it is not doable in this model. 
		As $l$ is the likelihood ratio of state $B$ over $A$ conditional on observing the $l$, for consistency, $\frac{p_B(l)}{p_A(l)}=l$ must be true. One can verify that  $p_B'(1)=p_A'(1)=0$ cannot both hold under this restriction. So we assume that $l_A,l_B$ are close to $1$ but are not $1$. This assumption also makes some proofs easier.}
		\item $\lim_{l\to 0^+}p_A(l)=\lim_{l\to +\infty}p_A(l)=\lim_{l\to 0^+}p_B(l)=\lim_{l\to +\infty}p_B(l)=0$.
	\end{enumerate}
\end{assumption}

We shall see that the growth rate of payoff function $P(I)$ and the vanishing rate of functions of success $p_A(l),p_B(l)$ jointly determine whether learning is complete in the long run. Thus, we need a way to measure and compare the growth and vanishing rate. 
Hardy's L-function provides us the necessary tool. (see appendix \ref{appendix1} for details). 
\begin{assumption}
	We only consider those payoff functions $\mathbf{P}$ which behave like a L-function as $I\to +\infty$.
	
	That is, for each $P(I)\in \mathbf{P}$, there exists a L-function $L(I)$ such that $\lim_{I\to +\infty} P(I)/L(I)=1$. 
\end{assumption}
Among all the benefits brought by the $L$-functions, the most important one is that the set of payoff functions $\mathbf{P}$ is totally order by the growth rate: 
\begin{lemma}
	Given two functions $P_1,P_2\in \mathbf{P}$, we write $P_1\preceq P_2$ iff $\lim_{I\to +\infty}\frac{P_1(I)}{P_2(I)}\leq 1$
	\footnote{
	Here $P_1\prec P_2$ means that $\lim_{I\to +\infty}\frac{P_1(I)}{P_2(I)}< 1$. It may worth to mention that this order $\prec$ is not the same as the prevailing orders used to study $L$-functions. In \cite{Hardy1910}, Hardy used $P_1\prec P_2$ to mean that $\lim_{I\to +\infty}\frac{P_1(I)}{P_2(I)}=0$. Often, $P_1\prec P_2$ is also used to mean that $P_1(I)<P_2(I)$ holds for all sufficiently large $I$.
}
	. Then one of the following three relations must hold between any two functions in $\mathbf{P}$.
	\begin{eqnarray}
		P_1\prec P_2, P_1\sim P_2, P_1\succ P_2.
	\end{eqnarray}
\end{lemma}
Similarly,	we also assume that $p_A(l),p_B(l)$ behave like a L-function in the tails. 
That is, 
\begin{assumption}
(1)$\exists$ L-function $L_1(l)$ such that $\lim_{l\to +\infty} p_A(l)/L_1(l)=1$; (2)$\exists$ L-function $L_2(l)$ such that $\lim_{l\to 0^+} p_A(l)/L_2(1/l)=1$.	
	Similar for $p_B$.
\end{assumption}

Lastly, we impose some technical assumptions.
\begin{assumption}
	(1) $P(0)=c>0$. That is, despite that different payoff functions provide different incentives for information, they all pay the same base salary for zero information.
	(2) $p_A(l),p_B(l)$ and $P(I)$ are all continuously differentiable.
\end{assumption}

\section{Model Solution}
\label{section3}
Our first job is to show that each researcher's individual optimization problem 
\begin{eqnarray}
	\label{eqn2.0}
	\max_{l\in (0,+\infty)} P(I(u,l))[up_A(l)+(1-u)p_B(l)]
\end{eqnarray}
has solutions for each $u$ despite the constraint set is open. This relies on the following special property of $I(u,l)$.

\begin{lemma}
	\label{lemma3.1}
	For each $u\in (0,0.5)$, 
	\begin{eqnarray}
		\forall l \in (0,+\infty), I(u,l)<\lim_{l\to 0} I(u,l)=-\log u.
	\end{eqnarray}	
	Similarly, for each $u\in [0.5,1)$,
	\begin{eqnarray}
		\forall l \in (0,+\infty), I(u,l)<\lim_{l\to +\infty} I(u,l)=-\log (1-u).
	\end{eqnarray}
\end{lemma}
This lemma says two things. First, shift all the weight to the underweighted state generates the supremum of information. For example, if current belief $u$ underweight state $A$ ($u<0.5$), and we could use $l=0$ to shift all the weight to $A$, the associated information $I(u,0)=-\log u$ is larger than any achievable information. Second, the supremum of information is finite. 

This implies the obtainable payment $P(I(u,l))$ is bounded from above for each $u$. Thus, we could safely ignore extreme $l$ as $p_A(l),p_B(l)$ vanishes there. Nothing is lost if we optimize on a large enough compact set. In fact, we have
\begin{proposition}
	\label{prop2.1}
	There exists a compact-valued, continuous correspondence $D(u)$ defined on $u\in (0,1)$, such that the original optimization problem \ref{eqn2.0} is equivalent to 
	\begin{eqnarray}
		\max_{l\in D(u)} P(I(u,l))[up_A(l)+(1-u)p_B(l)]
	\end{eqnarray}
for each $u\in (0,1)$.
\end{proposition}
As a result of continuous maximum theorem, we have the following corollary
\begin{corollary}
	The solution correspondence of \ref{eqn2.0}, $l^*(u)$, is compact-valued and u.h.c. Besides, the optimal value function $EP(u,l^*(u))$ is continuous in $u$.
\end{corollary}
To abbreviate notation, here we use $EP(u,l)$ to stand for $P(I(u,l))[up_A(l)+(1-u)p_B(l)]$, which is the expected payoff under $P(I)$.

\medskip

In analyzing the long run behavior of beliefs, it is easier to use the log likelihood ratio $\lambda_t^\varepsilon=\log \frac{1-u_t^\varepsilon}{u_t^\varepsilon}$. We say that learning is complete iff $\lambda_t^\varepsilon\to -\infty$. 

The expected change of $\lambda_t^\varepsilon$ is
\begin{eqnarray}
	\label{eqn2.1}
	E_t[\lambda_{t+1}^\varepsilon-\lambda_t^\varepsilon|\lambda_t^\varepsilon]=p(p_A(l^*)+\varepsilon)\log l^*+(1-p(p_A(l^*)+\varepsilon))\log \frac{1-pp_B(l^*)}{1-pp_A(l^*)}.
\end{eqnarray}
Rearrange the terms, we can separate the expected change into two terms
\begin{eqnarray}
	\label{eqn2.2}
	&&E_t[\lambda_{t+1}^\varepsilon-\lambda_t^\varepsilon|\lambda_t^\varepsilon]\notag\\
	&=&\Big(pp_A(l^*)\log l^*+(1-pp_A(l^*))\log \frac{1-pp_B(l^*)}{1-pp_A(l^*)}\Big)
	+\varepsilon p\log \frac{l^*(1-pp_A(l^*))}{1-pp_B(l^*)}
\end{eqnarray}
The first term is the expected change without p-hacking $E[\lambda_{t+1}^0-\lambda_{t+1}^0|\lambda_t^0]$. The second term is the distortion term associated with p-hacking.
Using Jensen's inequality, it is direct to verify that $E[\lambda_{t+1}^0-\lambda_{t+1}^0|\lambda_t^0]<0$.
It is natural to ask, if the intensity $\varepsilon$ is sufficiently small, is the distortion term also sufficiently small so that $E_t[\lambda_{t+1}^\varepsilon-\lambda_t^\varepsilon|\lambda_t^\varepsilon]<0$ and $\lambda_t^\varepsilon$ is a supermartingale that converges to $-\infty$?

The following lemma gives a sufficient condition for $\lambda_t^\varepsilon$ to be supermartingale.
\begin{lemma}
	\label{lemma2.1}
	If the set of all projects that could be chosen-$\{l^*(u)|u\in (0,1)\}\equiv L^*$-is bounded away from $0$ and $+\infty$ ($L^*$ constricted), then for any sufficiently small $\delta$, 
	$\exists c>0$, such that $E[\lambda_{t+1}^\varepsilon-\lambda_t^\varepsilon|\lambda_t^\varepsilon]\leq -\delta <0$ for all $\varepsilon<c$. 
	
	In other words, as the intensity of p-hacking is small, process $\lambda_t^\varepsilon$ is a supermartingale.
\end{lemma} 
\begin{proof}
	First, $E[\lambda^0_{t+1}-\lambda^0_t|\lambda^0_t]$ is a function of $l^*$ which is defined and continuous on $(0,+\infty)$. Using Jensen's inequality, we conclude that $E[\lambda^0_{t+1}-\lambda^0_t|\lambda^0_t]<0$ on $l^*\in (0,+\infty)$. As a result, for any compact subset $S$ in $(0,+\infty)$, $E[\lambda^0_{t+1}-\lambda^0_t|\lambda^0_t]\leq s<0$ for all $l^*\in S$. Since $L^*$ is constricted, it is contained in a compact subset. Thus $E[\lambda^0_{t+1}-\lambda^0_t|\lambda^0_t]$ is strictly bounded below $0$ for all $l^*\in L^*$.
	
	Similarly, the term $\log \frac{l^*(1-pp_A(l^*))}{1-pp_B(l^*)}$ is also defined and continuous on $l^*\in (0,+\infty)$. It is positive only if $l^*>1$. As $L^*$ is constricted, the term must be bounded from above. Thus, the distortion term is small provided that $\varepsilon$ is sufficiently small. 
	
	The lemma follows directly. 
\end{proof}

Not every supermartingale converges to $-\infty$. But in our case, we could use martingale central limit theorem to show that $\lambda_t^\varepsilon$ must converge to $-\infty$ almost surely. A one sentence summarization of the proof is that $\lambda_t^\varepsilon$ is a supermartingale whose drift dominates its variance. Interested readers could refer to appendix \ref{appendixB} for details.
\begin{proposition}
	\label{thm2.2}
	If $L^*$ is constricted, then as long as the intensity of p-hacking $\varepsilon$ is sufficiently small, then $\lambda_t^{\varepsilon}\to -\infty$ almost surely, that is,
	learning is complete a.s.
\end{proposition} 

So, could we expect that $L^*$ is always constricted? A little thought casts some doubt. If the current belief has assigned enough weight to some state, then pushing beliefs further to this state generates almost no information. In other words, researchers naturally have the incentive to go contrarian when belief is extreme. The only thing that could prevent them from going contrarian is that the success probability of a strong contrarian evidence is small. 
But if $P(I)$ grows fast enough, the small success probability can be compensated by the large reward upon success, and we can no longer expect $L^*$ to be constricted. What would happen in this case?

We have the following theorem
\begin{theorem}
	\label{thm2.7}
	If the payment function $P(I)$ grows fast enough, then the optimal project $l^*$ goes to $+\infty$ as $\lambda^\varepsilon_t$ goes to $-\infty$. In this case, as long as the intensity of p-hacking is not $0$, the event that learning is complete happens with probability $0$. 
		
	In other words, under fast growing payoff functions, learning is never complete as long as there is p-hacking.
\end{theorem}
We delay the analysis of the growth rate of payoff functions and the constrictedness of $L^*$ into the next section. That 
\begin{eqnarray}
	\label{eqn2.11}
	\lim_{\lambda_t^\varepsilon\to -\infty} l^*(\lambda_t^\varepsilon)=+\infty
\end{eqnarray}
implies complete learning never happens can be seen simply. For any $\varepsilon>0$, \ref{eqn2.11} implies that $E_t[\lambda_{t+1}^\varepsilon-\lambda_t^\varepsilon|\lambda_t^\varepsilon]$ are eventually positive for sufficiently negative $\lambda_t^\varepsilon$. Intuitively, this prevent $\lambda_t^\varepsilon$ from going to $-\infty$. A rigorous proof can be found in the appendix \ref{appendixC}.

The sibling of theorem \ref{thm2.7} is also true. 
\begin{theorem}
	\label{thm2.8}
	If the payment function $P(I)$ grows slowly enough, then $L^*$ is constricted. Learning is complete a.s. if the intensity of p-hacking is low.
\end{theorem}
Here the small success probability of going contrarian cannot be compensated because payment function $P(I)$ grows slowly. Detailed analysis can be found in the next section.

\section{Constrictedness of $L^*$}
\label{constrictedness}
As stated in the last section, researchers are incentivized to go contrarian when the belief assigns enough weight to one state and the payoff function grows fast enough. In this section we elaborate this idea.

We start our analysis with the researchers' decision problems when beliefs are getting super confident about state $A$ ($u_t\to 1$). 
A simple while important observation is that the information associated with any given project $l$ vanishes as $u\to 1$. The simplicity of this observation is revealed by the following fact. When the current belief assigns $90\%$ weight to state $A$, a project $l=10$ is sufficient to change the weight of $A$ to $47.4\%$ and making state $B$ more plausible. However, when the current belief assigns $99.9\%$ weight to state $A$, the same project only changes the weight of $A$ to $99\%$, which sounds not surprising at all. In fact, however large $l$ is, there is a $u$ close enough to $1$ so that the success of project $l$ only shift a tiny weight away from state $A$. As a result, in the process $u_t\to 1$, if researchers still want to generate some information, they not only need to go contrarian, they also need to go ``extremely" contrarian. If they want to stay safe by keeping the chosen project below some threshold $M$, eventually they generate zero information and only get paid by the base salary upon success. Since the payments upon success are just the base salary, researchers would choose the project with largest success probability, that is, project $l_A$. 
This statement can be written mathematically as the following lemma 
\begin{lemma}
	Let $u_k\to 1, l_k\in l^*(u_k)$ and $\{l_k\}\subset (0,M)$ with $M<+\infty$, then 
	\begin{eqnarray}
		l_k\to l_A \mbox{ and } \lim_{k\to +\infty} EP(u_k,l_k)=cp_A(l_A).
	\end{eqnarray}
is true for any payoff function $P$.
\end{lemma}
Since staying safe would generate payoffs about $cp_A(l_A)$, any strictly higher payoffs can only be obtained by going contrarian. We could explicitly construct a payoff function $P_1(I)$ with higher payoffs, so going contrarian must be more attractive under $P_1(I)$.
\begin{lemma}
	There exists a payoff function $P_1(I)$ such that 
	\begin{eqnarray}
		\lim_{u\to 1}EP_1(u,l^*(u))\geq \lim_{u\to 1} EP_1(u,(1-u)^{-1})>c.
	\end{eqnarray}
Thus, under $P_1(I)$
\begin{eqnarray}
	\lim_{u\to 1} \min l^*(u)\to +\infty.
\end{eqnarray}
\end{lemma}
It is not surprising that if going contrarian is optimal under one payoff function, it must also be optimal under another payoff function which grows faster. Therefore, we have the following proposition:
\begin{proposition}
	\label{prop4.3}
	There exists a non-empty set of payoff functions $\overline{\mathbf{P}}_r$ such that 
	\begin{eqnarray}
		\lim_{u\to 1} \min l^*(u)\to +\infty
	\end{eqnarray}
under any $P$ in $\overline{\mathbf{P}}_r$. Furthermore, if $P$ is in $\overline{\mathbf{P}}_r$ and $Q\succ P$, then $Q$ is in $\overline{\mathbf{P}}_r$ as well.
\end{proposition}

On the other hand, if $P(I)$ grows so slow that $\lim_{I\to +\infty} P(I)$ is a finite number, then the risk of going extremely contrarian (vanishing success probability) can never be compensated by a sufficiently large payoff upon success. Thus, $L^*$ must be constricted under such a slow growing payoff function. Furthermore, it is not surprising that if staying safe is optimal under one payoff function, it must also be optimal under another payoff function which grows slower. So we have the following proposition
\begin{proposition}
	\label{prop4.4}
	There exists a non-empty set of payoff functions $\underline{\mathbf{P}}$ such that 
	$L^*$ is constricted 
	under any $P$ in $\underline{\mathbf{P}}$. Furthermore, if $P$ is in $\underline{\mathbf{P}}$ and $Q\prec P$, then $Q$ is in $\underline{\mathbf{P}}$ as well.
\end{proposition}

Omitted details can be found in appendix \ref{appendixD}. Similar conclusions and analysis holds when public beliefs are getting extremely confident about state $B$ ($u\to 0$). 

\section{Conclusion}
In this paper we study the long run consequence of p-hacking under the framework of mis-specified learning. We find that over-incentivizing information provision could harm learning in the presence of p-hacking. On the other hand, p-hacking at a low intensity won't affect long run learning if the incentive of information provision is proper.

There are several potential extensions of this model.
In the analysis in the main text, we assume the function of success probabilities is fixed and vary the payoff functions. It is reasonable to think that some areas may have different $p_A(l)$ from other areas. For example, it might be true that psychology, economics and clinical studies have $p_A(l)$s with thick tails
\footnote{
We say $p_A^1(l)$ has a thicker tail than $p_A^2(l)$ if $$\lim_{l\to 0\mbox{ or }+\infty} \frac{p_A^1(l)}{p_A^2(l)}>1.$$
The word ``tail" may be a poor choice as $p_A(l)$ is not a distribution.
}
than physics, chemistry and astronomy. The wide heterogeneity of people makes it is relative easier to find contrarian evidences. In our framework, a $p_A(l)$ with thicker tails makes p-hacking a more serious problem. As the success probability of finding extreme contrarian evidences is higher, people are more likely to go contrarian. This may be related to the fact that psychologists and economists care more about p-hacking than physicists and chemists.

Besides, in the case of theorem \ref{thm2.7}, extreme contrarian evidences regularly arrive when public belief has assigned enough weight to one state. This is rarely observed in scientific research. But it is occasionally observed on social media-public opinions strongly favor one fact, then suddenly an extreme contrarian evidence overturns everything. Sometimes this radical opinion shifts could go back and forth several times. Consider the easiness of a social media blogger to p-hack (make the evidence looks stronger than it actually is) and the high reward of becoming an influencer on social media, theorem \ref{thm2.7} may be a channel to explain the radical shifts of public opinions on social media.
\newpage
\appendix
\section{Hardy's L-function}
\label{appendix1}
In \cite{Hardy1910}, Hardy introduced L-functions.
	In short, L-functions 
	are real-valued one-variable function obtained via finitely many operations of $+,-,\times,\div,\sqrt[n]{}$ and finitely many applications of operators $\log()$ and $e^{()}$.
	To be specific, L-functions could be constructed in the following recursive way:
	\begin{enumerate}
		\item order-$0$ $L$-function: algebraic functions that variable $x$ and constants are connected via finitely many operations of $+,-,\times,\div,\sqrt[n]{}$ 
		\item order-$1$ $L$-function: $\log$ and exponential of order $0$ L-function; and those elementary order $1$ $L$-functions, together with order $0$ $L$ functions, connected by finitely many algebraic operations.
	\end{enumerate}  
	Repeat this procedure one can obtain order $n$ L-functions for all natural number $n$.
	
For example, $L(x)=x^{x^x}$ could be rewritten as 
\begin{eqnarray}
	e^{(\log x) e^{x \log x}}.
\end{eqnarray}
Here $x\log x$ is an order-$1$ L-function; applying operator $e^{()}$ to it, we obtain an order-$2$ L-function $e^{x \log x}$; multiplying this order-$2$ L-function to another order-$1$ L-function $\log x$ we obtain a different order-$2$ L-function; we apply operator $e^{()}$ again to get order-$3$ L-function $x^{x^x}$.

Hardy proved that any $L$-functions are eventually continuous, monotonic, of constant sign, and has a limit as $x$ approaches $+\infty$. Of course, here this limit is allowed to be $\pm \infty$. L-functions have many other desirable properties as well. First, it provides a rich enough way to measure a function's growth rate.
\begin{proposition}
	Using $l_n(x)$ to denote the $n$-th iteration of logarithm, and $e_n(x)$ to denote the $n$-th iteration of exponential. Then there exist $L$-functions $f_n$ grows slower than $l_n(x)$, that is,
	\begin{eqnarray}
		\lim_{x\to +\infty}\frac{f_n(x)}{l_n(x)}=0, \lim_{x\to +\infty}f_n(x)=+\infty.
	\end{eqnarray}
	There also exist $L$ functions $g_n(x)$ grows faster than $e_n(x)$, that is,
	\begin{eqnarray}
		\lim_{x\to +\infty} \frac{e_n(x)}{g_n(x)}=0.
	\end{eqnarray}
\end{proposition}
From the construction of L-functions, the reciprocal of a L-function is still a L-function. So L-functions also provides a rich way to measure vanishing rate of functions.

Furthermore, the quotient of two L-functions is still a L-function, so it must have a limit as $x\to +\infty$.
Thus,
we could make the set of payoff functions $\mathbf{P}$ totally ordered according to its growth rate.
\begin{proposition}
	For any two payoff functions $P, Q$ in $\mathbf{P}$, define $P\preceq Q$ as 
	\begin{eqnarray}
		\lim_{I\to +\infty}\frac{P(I)}{Q(I)}\leq 1.
	\end{eqnarray}
	Then $\preceq$ is a total order on $\mathbf{P}$.
\end{proposition}

\section{Almost surely convergence under constricted $L^*$}
\label{appendixB}
The intuition of the proof goes as following: 
Using Doob's decomposition, supermartingale $\lambda_t^\varepsilon$ is the sum of a martingale $M_t$ and a predictable process (drift) $A_t$
\begin{eqnarray}
	A_t=\sum_{m=1}^t E[\lambda_m-\lambda_{m-1}|\mathcal{F}_{m-1}].
\end{eqnarray} 
\footnote{Here the filtration is generated by the supermartingale $\lambda_t^\varepsilon$. That is, $\mathcal{F}_{t}=\sigma(\lambda^\varepsilon_0,\dots,\lambda^\varepsilon_t)$. It is also direct to verify that this filtration $\mathcal{F}_t$ is also generated by the associated martingale $M_t$, that is,
	$\mathcal{F}_{t}=\sigma(M_0,\dots,M_t)$ holds as well.}
Martingale central limit theorem (corollary 1 in \cite{Ouchti2005}) implies that the distribution of $\lambda_t^\varepsilon$ is eventually roughly a normal distribution centered at $E[A_t]$ with variance 
\begin{eqnarray}
	\sum_{m=1}^t E[(M_m-M_{m-1})^2|\mathcal{F}_{m-1}].
\end{eqnarray}
We could compute this variance and find the standard deviation of $\lambda_t^\varepsilon$ is at the order of $\sqrt{t}$. As shown in lemma \ref{lemma2.1}, the drift satisfying that 
\begin{eqnarray}
	E[\lambda_t^\varepsilon]=E[A_t]\leq -\delta t.
\end{eqnarray}
As the drift is of order $t$ and dominates the standard deviation, $\lambda_t^\varepsilon$ converges to $-\infty$ almost surely.

\subsection{Convergence in probability}
In this subsection, we prove a lemma that supermartingale $\lambda_t^{\varepsilon}$ converges to $-\infty$ in probability. We need this lemma to show that the variance of $\lambda_t^{\varepsilon}$ diverge to $+\infty$ a.s., which is a required condition of martingale CLT.

\begin{lemma}
	\label{lemmaB1}
	Under the condition of lemma \ref{lemma2.1}, supermartingale $\lambda_t^\varepsilon$ converges to $-\infty$ in probability.
\end{lemma}

We need the following Azuma's inequality to prove lemma \ref{lemmaB1}.
\begin{lemma}(Azuma's inequality)
	If $X_t$ is a supermartingale satisfying that 
	\begin{eqnarray}
		|X_{t}-X_{t-1}|\leq c_t,\, a.s.\, \forall t\in \{1,2,\dots\}
	\end{eqnarray}
	then for all $\epsilon>0$, we have 
	\begin{eqnarray}
		P(X_t-X_0>\epsilon)\leq e^{-\frac{\epsilon^2}{2\sum_{s=1}^t c^2_s}}
	\end{eqnarray}
\end{lemma}
\begin{proof}[Proof of lemma \ref{lemmaB1}]
	We'll apply Azuma's inequality to an auxiliary process constructed as following: let $\delta$ be as in lemma \ref{lemma2.1}, for each history $h_t$, define 
	\begin{eqnarray}
		B_t(h_t)=\lambda_t^\varepsilon(h_t)+\frac{\delta}{2}t.
	\end{eqnarray}
	It is direct to verify that $B_t$ is still a supermartingale as 
	\begin{eqnarray}
		E[B_{t+1}-B_t|h_t]=E[\lambda^\varepsilon_{t+1}-\lambda^\varepsilon_t+\frac{\delta}{2}|h_t]\leq -\frac{\delta}{2}.
	\end{eqnarray}
	Because $\lambda_{t+1}^\varepsilon-\lambda_t^\varepsilon$ takes either $\log l_t^*$ or $\log\frac{1-pp_B(l_t^*)}{1-pp_A(l_t^*)}$ and that $l_t^*$ is bounded away from $0$ and $+\infty$, there exists $d>0$
	\begin{eqnarray}
		|\lambda_{t+1}^\varepsilon-\lambda_t^\varepsilon|\leq d, \forall t.
	\end{eqnarray}
	It is immediate that 
	\begin{eqnarray}
		|B_{t+1}-B_t|\leq d+\frac{\delta}{2}
	\end{eqnarray} 
	Applying the Azuma's inequality to $B_t$, we have
	\begin{eqnarray}
		P(B_t\geq B_0+\frac{\delta t}{4})\leq e^{-\frac{\delta^2 t}{32 (d+\frac{\delta}{2})^2}}.
	\end{eqnarray}
	From the construction of $B_t$, it is obvious that 
	\begin{eqnarray}
		P(B_t\geq B_0+\frac{\delta t}{4})=P(\lambda_t^\varepsilon\geq \lambda_0^\varepsilon-\frac{\delta t}{4}).
	\end{eqnarray}
	So $\lambda_t^\varepsilon$ converges to $-\infty$ in probability.
\end{proof}

\subsection{Proof of proposition \ref{thm2.2}}
We have decompose $\lambda_t$ into the sum of martingale $M_t$ and drift $A_t$.
\begin{eqnarray}
	\lambda_t=M_t+A_t.
\end{eqnarray}
We could directly compute the martingale difference $M_t-M_{t-1}$ as
\begin{eqnarray}
	&&M_t-M_{t-1}=\lambda_t-\lambda_{t-1}-E[\lambda_t-\lambda_{t-1}|\mathcal{F}_{t-1}]\notag\\
	&=& 
	\begin{cases}
		[1-p(p_A(l^*_{t-1})+\varepsilon)]\log \frac{l^*_{t-1}(1-pp_A(l^*_{t-1}))}{1-pp_B(l^*_{t-1})}, \mbox{ with prob } p(p_A(l^*_{t-1})+\varepsilon)\\
		-p(p_A(l^*_{t-1})+\varepsilon)\log \frac{l^*_{t-1}(1-pp_A(l^*_{t-1}))}{1-pp_B(l^*_{t-1})}, \mbox{ with prob }
		[1-p(p_A(l^*_{t-1})+\varepsilon)]
	\end{cases}
\end{eqnarray}
Let 
\begin{eqnarray}
	\sigma_{t-1}^2=E[(M_t-M_{t-1})^2|\mathcal{F}_{t-1}].
\end{eqnarray}
We can compute it out as 
\begin{eqnarray}
	\sigma_{t-1}^2=[1-p(p_A(l^*_{t-1})+\varepsilon)]p(p_A(l^*_{t-1})+\varepsilon)\Big(\log \frac{l^*_{t-1}(1-pp_A(l^*_{t-1}))}{1-pp_B(l^*_{t-1})}\Big)^2
\end{eqnarray}
Since $l_{t-1}^*$ is bounded away from $0$ and $+\infty$, $\sigma_{t-1}^2$ is always bounded from above
\begin{eqnarray}
	\exists S\in \mathbb{R}\, s.t. \sigma^2_{t-1}\leq S,\, \forall t.
\end{eqnarray}
Here we cannot conclude that $\sigma_{t-1}^2$ is always bounded from below by a positive number, as $l_{t-1}^*$ could be $1$. However, we could show that $\sigma_{t-1}^2$ is bounded from below a.s. for infinitely many periods. As a result, $\sum_{t=0}^{+\infty} \sigma^2_{t}$ diverges to $+\infty$ almost surely. 

This follows from the observation that $l_{t}^*\to 1$ must happen with probability $0$. Assume otherwise, 
\begin{eqnarray}
	H_0=\{h|l_t^*(h)\to 1\} \mbox{ and } P(H_0)>0.
\end{eqnarray}
then for any $h\in H_0$, the FOC for maximization problem \ref{eqn1.1} says that
\begin{eqnarray}
	P'(I(u,l))\frac{u(1-u)\log l^*}{u+(1-u)l^*}p_A(l^*)+P(I(u,l))[(u+(1-u)l^*)p_A'(l^*)+(1-u)p_A(l^*)]=0\notag
\end{eqnarray}
holds for all $u=u_t(h),l^*=l_t^*(h)$.
Let $t\to +\infty$, using $l_t^*(h)\to 1$, we conclude that 
\begin{eqnarray}
	\label{eqnB19}
	\lim_{t\to +\infty} u_t(h)=1+\frac{p_A'(1)}{p_A(1)}
\end{eqnarray}
As $l_A\neq 1$, either $p_A'(1)>0$ or $p_A'(1)<0$ holds. In the first case, \ref{eqnB19} cannot hold. In the second case, \ref{eqnB19} implies that $u_t$ converges to an interior value with probability $P(H_0)$. This contradicts lemma \ref{lemmaB1}. Lastly, as $l_t^*\to 1$ a.s. does not happen, the event that $l_t^*$ is bounded away from $1$ for infinitely many periods happens almost surely. It is direct to verify that $l_t^*$ bounded away from $1$ implies that $\sigma_{t-1}^2$ is bounded from below by a positive number. So we proved that $\sum_{t=0}^{+\infty} \sigma^2_{t}$ diverges to $+\infty$ almost surely. 

For each $\nu\in \mathbb{N}$, let $\tau_\nu$ be the first time that sum of $\sigma_t^2$ reaches $\nu$, that is,
\begin{eqnarray}
	\tau_\nu=\inf\{k\in \mathbb{N}| \sum_{t=0}^k\sigma_t^2\geq \nu\}.
\end{eqnarray}
Using the bounds of $\sigma_t^2$, we could derive a bound on $\tau_\nu$ as
\begin{eqnarray}
	\label{eqn10.8}
	\frac{\nu}{S}-1\leq \tau_\nu
\end{eqnarray}

It is easy to verify that $M_t$ is a martingale of bounded increment. 
Apply the martingale CLT, 
\begin{eqnarray}
	\label{eqn10.9}
	\sup_{x\in \mathbb{R}}|P(M_{\tau_\nu}\leq x\sqrt{v})-\Phi(x)|\leq \frac{c}{\nu^{\frac{1}{4}}}
\end{eqnarray}
Here $\Phi(x)$ is the distribution function of standard normal and $c$ is a constant that depend on the uniform bound of the martingale difference terms.

Recall that we have shown that $E[\lambda_t-\lambda_{t-1}|\mathcal{F}_{t-1}]\leq -\delta$ for sufficiently small $\varepsilon$. So $A_t\leq -\delta t$.
As $M_{\tau_\nu}=\lambda_{\tau_\nu}-A_{\tau_\nu}$, 
that $M_{\tau_\nu}\leq x\sqrt{\nu}$ implies that $\lambda_{\tau_\nu}\leq x\sqrt{\nu}-\delta \tau_\nu$, which further implies that $\lambda_{\tau_\nu}\leq x\sqrt{\nu}-\delta \frac{\nu}{S}+\delta$ using \ref{eqn10.8}. Thus, 
\begin{eqnarray}
	P(M_{\tau_\nu}\leq x\sqrt{\nu})\leq P(\lambda_{\tau_\nu}\leq x\sqrt{\nu} - \frac{\delta}{S}\nu+1)
\end{eqnarray}
Taking $x=\frac{\delta}{2S}\sqrt{\nu}$, then martingale CLT \ref{eqn10.9} implies that 
\begin{eqnarray}
	\Phi(\frac{\delta}{2S}\sqrt{\nu})-\frac{c}{\nu^{\frac{1}{4}}}\leq P(M_{\tau_\nu}\leq x\sqrt{\nu})\leq P(\lambda_{\tau_\nu}\leq - \frac{\delta}{2S}\nu+1).
\end{eqnarray}

Now let $\mu=\nu^8$ and consider the subsequence $\lambda_{\tau_\mu}$, then 
\begin{eqnarray}
	P(\lambda_{\tau_\mu}>-\frac{\delta}{2S}\nu+1)\leq P(\lambda_{\tau_\mu}>-\frac{\delta}{2S}\mu+1) \leq  \frac{c}{\nu^2}+1-\Phi(\frac{\delta}{2S}\nu^4)
\end{eqnarray}
Using the fact that $1-\Phi(x)\leq e^{-\frac{x^2}{2}}$ holds for $x>0$, we further obtains
\begin{eqnarray}
	P(\lambda_{\tau_\mu}>-\frac{\delta}{2S}\nu+1)\leq \frac{c}{\nu^2}+e^{-\frac{\delta^2}{8S^2}\nu^4}.
\end{eqnarray}
Using Borel-Cantelli's lemma, as 
\begin{eqnarray}
	\sum_{\nu=1}^{+\infty}(\lambda_{\tau_\mu}>-\frac{\delta}{2S}\nu+1)<+\infty
\end{eqnarray}
we know that 
\begin{eqnarray}
	P(\lambda_{\tau_\mu}\leq -\frac{\delta}{2S}\nu+1 \mbox{ eventually})=1.
\end{eqnarray}
Thus, the subsequence $\lambda_{\tau_\mu}\to -\infty$ almost surely. This further implies that $\lambda_t\to -\infty$ almost surely as stopping times $\tau_\mu$ diverge to $+\infty$ a.s..

\section{Omitted part in the proof of Theorem \ref{thm2.7}}
\label{appendixC}
In this section we prove the following proposition
\begin{proposition}
	\label{propC1}
	If $\lim_{\lambda_t^\varepsilon\to -\infty} l^*(\lambda_t^\varepsilon)=+\infty$, then the event that $\lambda_t^\varepsilon\to -\infty$ happens with probability $0$.
\end{proposition}
This together with proposition \ref{prop4.3} prove theorem \ref{thm2.7}.

\begin{proof}[Proof of proposition \ref{propC1}]
If $l_t^*(\lambda_t^\varepsilon)\to +\infty$ as $\lambda_t^\varepsilon\to -\infty$, then 
\begin{eqnarray}
	&&\lim_{\lambda_t^\varepsilon\to-\infty}E_t[\lambda_{t+1}^\varepsilon-\lambda_t^\varepsilon|\lambda_t^\varepsilon]\notag\\
	&=&\lim_{l^*\to +\infty}\Big(pp_A(l^*)\log l^*+(1-pp_A(l^*))\log \frac{1-pp_B(l^*)}{1-pp_A(l^*)}\Big)
	+\varepsilon p\log \frac{l^*(1-pp_A(l^*))}{1-pp_B(l^*)}\notag\\
	&=&+\infty
\end{eqnarray}
In other words, the expected change of $\lambda_t^\varepsilon$ turns positive when $\lambda_t^\varepsilon$ is small enough. This intuitively justifies that the event that $\lambda_t^\varepsilon$ happens with probability zero.

The rigorous proof is given as following. 
Let 
\begin{eqnarray}
	H=\{h|\lim_{t\to+\infty} \lambda_t^\varepsilon(h)=-\infty\}.
\end{eqnarray}
That is, $H$ is the set of histories along which complete learning happens. We would break $H$ into a countable union of sets. Let $\bar{\lambda}$ be as defined in lemma \ref{lemmaC1}. For any period $s$, let
\begin{eqnarray}
	H_s=\{h|\lambda_t^\varepsilon<\overline{\lambda}, \forall t\geq s\mbox{ and }
	\lim_{t\to+\infty} \lambda_t^\varepsilon(h)=-\infty
	\}.
\end{eqnarray}
In other words, $H_s$ is the set of histories along which complete learning happens and the belief $\lambda_t^\varepsilon$ stays strictly below $\bar{\lambda}$ from period $s$ on.
It is direct to verify that 
\begin{eqnarray}
	H=\bigcup_{s\in \mathbb{N}}H_s.
\end{eqnarray}
The set $H_s$ can be further broken into a countable union. As in each period, the chosen project $l^*$ could either succeed or fail, there are $2^{s-1}$ different histories at period $s$. Let $\bar{h}_s$ be the set of histories $h_s$ that satisfy the condition that $\lambda_s^\varepsilon(h_s)<\bar{\lambda}$. For these $h_s$, we define
\begin{eqnarray}
	H^{h_s}=\{h|\lambda_s^\varepsilon<\overline{\lambda}, \forall t\geq s\mbox{ and }
	\lim_{t\to+\infty} \lambda_t^\varepsilon(h)=-\infty, \mbox{ and first s periods of } h \mbox{ is } h_s
	\}.
\end{eqnarray}
It is direct to verify that 
\begin{eqnarray}
	H_s=\bigcup_{h_s\in \overline{h}_s} H^{h_s}.
\end{eqnarray}
Then, if $P(H)>0$, there must exists at least one $h_{s_0}$ such that $P(H^{h_{s_0}})>0$. We shall show that this would lead to a contradiction.

Conditional on history $h_{s_0}$, define stopping time
\begin{eqnarray}
	\tau=\inf \{t>s_0|\lambda_t^\varepsilon\geq \bar{\lambda}\}.
\end{eqnarray}
to be the first time that $\lambda_t^\varepsilon$ escaping the region $(-\infty,\bar{\lambda})$. Construct auxiliary process $X_t$
\footnote{Here $X_t$ are defined conditional on history $h_{s_0}$. That is, it is defined for history $h$ whose first $s_0$ periods agrees with $h_{s_0}$.}
in the following way:
\begin{enumerate}
	\item if $t<\tau(h)$, then 
	\begin{eqnarray}
		X_{t+1}(h)-X_t(h)=
		\begin{cases}
			\log\overline{l};\mbox{if along } h \mbox{ project } l_t^* \mbox{ succeeds}\\
			\log\frac{1-pp_B(l^*_t(h))}{1-pp_A(l^*_t(h))};\mbox{if along } h \mbox{ project } l_t^* \mbox{ fails}
		\end{cases}
	\end{eqnarray}
\item if $t\geq \tau(h)$, then $X_{t+1}(h)-X_t(h)=0$.
\end{enumerate}
In other words, before $\lambda_t^\varepsilon$ first escape region $(-\infty,\bar{\lambda})$, $X_t$ evolves in the same way as $\lambda_t^\varepsilon$ does. The only difference is that $\lambda_t^\varepsilon$ moves up by value $\log(l_t^*(\lambda_t^\varepsilon))$, while $X_t$ moves up by a smaller value $\log\overline{l}$
\footnote{
As $\lambda_t^\varepsilon$ is in $(-\infty,\bar{\lambda})$, $l_t^*>\overline{l}$ by lemma \ref{lemmaC1}.
}
. Once $\lambda_t^\varepsilon$ escape $(-\infty,\bar{\lambda})$, $X_t$ stop evolving. 
It is immediate to verify that 
\begin{eqnarray}
	X_t\leq \lambda_{t\wedge \tau}^\varepsilon.
\end{eqnarray}
That is, $X_t$ stays below the stopped process $\lambda_{t\wedge \tau}^\varepsilon$ for all histories.

We could verify that $X_t$ is a submartingale. The expected change of $X_t$ conditional on history $h_t,t\geq s_0$ could take two values:
\begin{enumerate}
	\item if $t<\tau(h_t)$,
	then 
	\begin{eqnarray}
		&&E[X_{t+1}(h)-X_t(h)|h_t]\notag\\
		&=&p(p_A(l^*(h_t))+\varepsilon)\log \overline{l}+[1-p(p_A(l^*(h_t))+\varepsilon)]\log\frac{1-pp_B(l^*(h_t))}{1-pp_A(l^*(h_t))}
	\end{eqnarray} 
This is positive by lemma \ref{lemmaC1}.
	\item if $t<\tau(h_t)$, then $E[X_{t+1}(h)-X_t(h)|h_t]=0$.
\end{enumerate}
Obviously that $X_t$ is of bounded increment, so $\sup_{t} X_t^+<+\infty$. Thus, we could apply submartingale convergence theorem to conclude that $X_t$ converges almost surely to a limit random variable $X_\infty$ with $E|X_\infty|<+\infty$.
\footnote{Note that we cannot directly apply submartingale convergence theorem to the stopped process $\lambda_{t\wedge \tau}^\varepsilon$, as this stopped process needs not to be of bounded increment and hence $\sup_{t} X_t^+<+\infty$ needs not to be true.}
That $E|X_\infty|<+\infty$ implies that 
\begin{eqnarray}
	P(\{h|\lim_{t\to +\infty} X_t(h)=-\infty\}|h_{s_0})=0.
\end{eqnarray}
That is, conditional on history $h_{s_0}$, the event that $X_t$ goes to $-\infty$ happens with probability $0$. As $X_t\leq \lambda_{t\wedge \tau}^\varepsilon$, the event that $\lambda_{t\wedge \tau}^\varepsilon$ goes to $-\infty$ must happen with probability $0$, conditional on history $h_{s_0}$.

However, for any history $h\in H^{h_{s_0}}$, we must have
\begin{eqnarray}
	\lim_{t\to +\infty}\lambda_{t\wedge \tau}^\varepsilon(h)=-\infty.
\end{eqnarray}
Thus, the assumption that $P(H^{h_{s_0}})>0$ contradicts the conclusion in the last paragraph.
\end{proof}

\subsection{Some notations}
We include the following statement to nail down necessary notations.
\begin{lemma}
	\label{lemmaC1}
	For any $\delta$ be a sufficiently small positive number, there exists a $\bar{\lambda}<0$
	such that 
	\begin{eqnarray}
		&&\lambda_t^\varepsilon<\bar{\lambda}\notag\\
		&\Rightarrow& l^*(\lambda_t^\varepsilon)>\overline{l}>1\notag\\
		&\Rightarrow& p(p_A(l^*)+\varepsilon)\log \overline{l}+[1-p(p_A(l^*)+\varepsilon)]\log\frac{1-pp_B(l^*)}{1-pp_A(l^*)}\geq \varepsilon p \log\overline{l}-\delta>0
	\end{eqnarray}
\end{lemma}
\begin{proof}
	As $\lim_{l\to +\infty}p_A(l)=\lim_{l\to +\infty}p_B(l)=0$, $\Big|\log\frac{1-pp_B(l^*)}{1-pp_A(l^*)}\Big|\to 0$. Thus, there exists $\overline{l}>1$ such that the last inequality holds for all $l>\overline{l}$. The existence of $\bar{\lambda}$ comes from the assumption that $\lim_{\lambda_t^\varepsilon\to -\infty}l^*(\lambda^\varepsilon_t)=+\infty$.
\end{proof}

\section{Other omitted proofs}
\label{appendixD}
\subsection{Omitted proof in section \ref{section3}}
Lemma \ref{lemma3.1} is a direct consequence of corollary \ref{coro3.60}.

\begin{proof}[Proof of proposition \ref{prop2.1}]
	We explicitly construct the correspondence $D(u)$ as following.
	
	Choose $l_1$ to be a constant satisfying that $l_1>\max{l_B,1}$, and $l_2$ be a constant satisfying that $l_2<\min\{1,l_A\}$. Let $m(u)=\{-\log u,-\log(1-u)\}$. We have the following claim:
	\begin{claim}
		For each $u\in (0,1)$, 
		\begin{eqnarray}
			P(m(u))p_B(\overline{l})=EP(u,l_1)
		\end{eqnarray}
		uniquely determines a $\overline{l}\in (l_1,+\infty)$ and so-defined function $\overline{l}(u)$ is continuous; similarly,
		\begin{eqnarray}
			P(m(u))p_A(\underline{l})=EP(u,l_2)
		\end{eqnarray}
		uniquely determines a $\underline{l}\in (0,l_2)$ and so-defined function $\underline{l}(u)$ is continuous
	\end{claim}
	The correspondence $D(u)$ would be $[\underline{l}(u),\overline{l}(u)]$.
	
	We start with the proof of the claim. First, fixing $u$, then for each $l>1$, we have that 
	\begin{eqnarray}
		EP(u,l)=P(I(u,l))[up_A(l)+(1-u)p_B(l)]<P(I(u,l))p_B(l)\leq P(m(u))p_B(l).
	\end{eqnarray}
	Thus, $P(m(u))p_B(l_1)>EP(u,l_1)$. On the other hand, $\lim_{l\to +\infty}P(m(u))p_B(l)=0<EP(u,l_1)$. By continuity, there must exist a $\overline{l}(u)\in (l_1,+\infty)\subset (l_B,+\infty)$ such that the equality holds. Furthermore, since $p_B(l)$ strictly decreases on $(l_B,+\infty)$, such $\overline{l}$ must be unique. Now let $u_k\to \overline{u}\in (0,1)$, by continuity of $EP(u,l_1)$ and $P(m(u))$,  $p_B(\overline{l}(u_k))=\frac{EP(u_k,l_1)}{P(m(u_k))}$ must converge to $\frac{EP(\overline{u},l_1)}{P(m(\overline{u}))}$
	. That is, $\lim_{k\to +\infty} p_B(\overline{l}_k(u_k))= p_B(\overline{l}(\overline{u}))$.
	As $\{\overline{l}_k(u_k)\},\overline{l}(\overline{u})\subset (l_1,+\infty)$ and $P_B(l)$ strictly decreases on $(l_1,+\infty)$, we must have $\lim_{k\to +\infty} \overline{l}_k(u_k)= \overline{l}(\overline{u})$.
	
	A similar argument works for the second part of the claim. First, fixing $u$, then for each $l<1$,
	\begin{eqnarray}
		EP(u,l)=P(I(u,l))[up_A(l)+(1-u)p_B(l)]<P(I(u,l))p_A(l)\leq P(m(u))p_A(l).
	\end{eqnarray}
	Thus $P(m(u))p_A(l_2)>EP(u,l_2)$. On the other hand, $\lim_{l\to 0} P(m(u))p_A(l)=0<EP(u,L_2)$. So by continuity, there must exist $\underline{l}(u)\in (0,l_2)$ such that the equality holds. Furthermore, since $p_A(l)$ strictly increases on $(0,l_2)$, such a $\underline{l}(u)$ must be unique. The argument for the continuity of $\underline{l}(u)$ is exactly the same as in the first part.
	
	From the construction, $D(u)=[\underline{l}(u),\overline{l}(u)]$ is a compact-valued, continuous correspondence. Furthermore, we must have 
	\begin{eqnarray}
		&&EP(u,l)<P(m(u))p_A(l)<P(m(u))p_A(\underline{l}(u))=EP(u,l_2), \forall l<\underline{l}(u)\notag\\
		&&EP(u,l)<P(m(u))p_B(l)<P(m(u))p_B(\overline{l}(u))=EP(u,l_1), \forall l>\overline{l}(u)\notag
	\end{eqnarray}
	Thus, any $l\notin D(u)$ cannot be optimal.
\end{proof}

\subsection{Omitted proof in section \ref{constrictedness}}
\begin{lemma}
	\label{lemma3.5}
	Let $u_k\to 1, l_k\in l^*(u_k)$ and $\{l_k\}\subset (0,M)$ with $M<+\infty$, then 
	\begin{eqnarray}
		l_k\to l_A \mbox{ and } \lim_{k\to +\infty} EP(u_k,l_k)=cp_A(l_A).
	\end{eqnarray}
	Besides, it is direct to verify that 
	\begin{eqnarray}
		u_k\to 1\Rightarrow EP(u_k,l_A)\to cp_A(l_A).
	\end{eqnarray}
\end{lemma}
\begin{proof}
	As $l_k\subset (0,M)$, $\exists$ subsequence $k_n$ such that $l_{k_n}$ converges to $\overline{l}\in [0,M]$. 
	
	We first show that $\overline{l}$ cannot be $0$. Otherwise, $I_{k_n}=I(u_{k_n},l_{k_n})\to 0$. Then $EP(u_{k_n},l_{k_n})\to 0$. On the other hand, fixing choice to be $1$, $EP(u_{k_n},1)=cp_A(1)>0$. So $l_{k_n}\to 0$ cannot be optimal.
	
	Now that $0$ cannot be a cluster point of $\{l_k\}$. We know $\overline{l}\in [m,M]$ with that $0<m<M<+\infty$.	
	As $l_{k_n}$ is optimal, $(u_{k_n},l_{k_n})$ satisfy FOC 
	\begin{eqnarray}
		P'(I(u_{k_n},l_{k_n}))\frac{u_{k_n}(1-u_{k_n})\log l_{k_n}}{u_{k_n}+(1-u_{k_n})l_{k_n}}p_A(l_{k_n})+P(I(u_{k_n},l_{k_n}))[u_{k_n}p_A'(l_{k_n})+(1-u_{k_n})p_B'(l_{k_n})]=0.\notag
	\end{eqnarray}
	Let $k_n\to +\infty$, since $l_{k_n}, \log l_{k_n}$ are bounded, the first term goes to $0$. Using lemma \ref{lemma3.80}, the second term goes to $P(0)p'_A(\overline{l})$, which must equal $0$. Since $P(0)\neq 0$, it must be the case that $p'_A(\overline{l})=0$. By our assumption \ref{assumption1.1}, $\overline{l}=l_A$. Using lemma \ref{lemma3.80} again, $EP(u_k,l_k)\to cp_A(l_A)$ is immediate. 
\end{proof}\\
There is a similar result which can be proved similarly
\begin{lemma}
	\label{lemma3.6}
	Let $u_k\to 0, l_k\in l^*(u_k)$ and $\{l_k\}\subset (m,+\infty)$ with $0<m<+\infty$, then 
	\begin{eqnarray}
		l_k\to l_B \mbox{ and } \lim_{k\to +\infty} EP(u_k,l_k)=cp_B(l_B).
	\end{eqnarray}
	Besides, it is direct to verify that 
	\begin{eqnarray}
		u_k\to 0\Rightarrow EP(u_k,l_B)\to cP_B(l_B).
	\end{eqnarray}
\end{lemma}

\begin{lemma}
	\label{lemma3.7}
	There exists a payoff function $P_1(I)$ such that 
	\begin{eqnarray}
		\lim_{u\to 1} EP_1(u,(1-u)^{-1})>c.
	\end{eqnarray}
\end{lemma}
\begin{proof}
	If $u\to 1$, we want the expected payoff associated with $(u,(1-u)^{-1})$  
	\begin{eqnarray}
		P_1\Big(-\frac{\log(1-u)}{1+u}-\log(u+1)\Big)p_A\Big(\frac{1}{1-u}\Big)(u+1).
	\end{eqnarray}
	to be eventually larger than the base salary $c$. 
	This holds as long as 
	\begin{eqnarray}
		\label{eqn1-32-2023-10-20}
		\lim_{x\to +\infty}p_A(x)P_1\Big(\frac{\log x}{2}-\log 2\Big)> c.
	\end{eqnarray}
	This can be done by choosing $P_1(x)=d[p_A(4e^{2x})]^{-1}$ with $d=\frac{c}{p_A(4)}>c$. 
	\footnote{As we assume that $p_A(l_A)$ peaks at $l_A$ near $1$ and strictly decreases afterwards, $p_A(4)<1$.} 
	This special choice of $d$ guarantees that the constructed $P_1(I)$ also has base salary $c$.
	
	Furthermore, since $\lim_{x\to +\infty}p_A(x)/L_1(x)=1$ for some L-function, 
	$$\lim_{x\to +\infty} P_1(x)/d[L(4e^{2x})]^{-1}=1.$$ And $d[L(4e^{2x})]^{-1}$ is also a L-function. Lastly, as $p_A(x)$ strictly decreases for $x>4$, $P_1(x)=d[p_A(4e^{2x})]$ strictly increases for $x\geq 0$. Thus, the constructed $P_1(x)$ satisfies all the assumptions of a payoff function.
\end{proof}\\
Obviously, $P_1(I)\in \overline{\mathbf{P}}_r$.
Now we turn to prove the non-emptiness of $\underline{\mathbf{P}}$.
\begin{proposition}
	Let $P(I)$ satisfies that $\lim_{I\to +\infty} P(I)<+\infty$. Then $P(I)$ is in  $\underline{\mathbf{P}}$. 
\end{proposition}
\begin{proof}
	Assume otherwise, under $P(I)$, $L^*$ is not constricted from the above. By lemma \ref{lemma3.270}, there exists $u_k\to 1$ and $l_k\in l^*_P(u_k)$ satisfying that $l_k\to +\infty$. Since $\lim_{I\to +\infty} P(I)<+\infty$, $P(I_k)$ is bounded away from $+\infty$ even if $I_k\to +\infty$. Thus, $EQ(u_k,l_k)\to 0$. By lemma \ref{lemma3.5}, $EQ(u_k,l_A)\to cp_A(l_A)>0$. This contradicts the optimality of $l_k$.
	
	In a similar way, we could prove that $L^*$ under $P(I)$ must be constricted from below.
\end{proof}

Now we prove the second part of proposition \ref{prop4.3}.
\begin{proposition}
	\label{prop3.10}
	$P\in \overline{\mathbf{P}}_r$ and $Q\succ P$ implies that $Q\in \overline{\mathbf{P}}_r$
\end{proposition}
\begin{proof}
Assume otherwise, $Q\notin \overline{\mathbf{P}}_r$. Then there exists $u_k\to 1, l_k\in l^*_Q(u_k)$ and $\{l_k\}\subset (0,M)$ uniformly bounded from above.
		
		By lemma \ref{lemma3.5}, 
		\begin{eqnarray}
			\label{eqn3.17}
			EQ(u_k,l_k)\to cp_A(l_A).
		\end{eqnarray}
		
		On the other hand, as $EP(u_k,l^*_P(u_k))\geq EP(u_k,l_A)$
		\begin{eqnarray}
			\liminf_{k\to +\infty} EP(u_k,l^*_P(u_k))\geq cp_A(l_A).
		\end{eqnarray}
		Furthermore, $I_k=I(u_k,l^*_P(u_k))$ must go to $+\infty$. Otherwise, there exists a subsequence $I_{k_n}$ goes to a finite $\overline{I}$. But this means that the vanishing success probability cannot be compensated. Hence $EP(u_{k_n},l^*_P(u_{k_n}))\to 0$. 
		
		Now, by switching from $l_k$ to $l^*_P(u_k)$, we could guarantee a higher payoff
		\begin{eqnarray}
			&&\lim_{k\to +\infty} EQ(u_k,l_k)\geq \liminf_{k\to +\infty} EQ(u_k,l^*_P(u_k))\notag\\
			&=&\liminf_{k\to +\infty} Q(I(u_k,l^*_P(u_k)))[u_kp_A(l^*_P(u_k))+(1-u_k)p_B(l^*_P(u_k))]\notag\\
			&\geq& a \liminf_{k\to +\infty} P(I(u_k,l^*_P(u_k)))[u_kp_A(l^*_P(u_k))+(1-u_k)p_B(l^*_P(u_k))]\notag\\
			&\geq& acp_A(l_A) \mbox{ with } a>1.
		\end{eqnarray}
		This obviously contradicts limit in \ref{eqn3.17}.
\end{proof}

Now we prove the second part of proposition \ref{prop4.4}
\begin{proposition}
	\label{prop3.11}
	$Q\in \underline{\mathbf{P}}$ and $P\prec Q$ implies that $P\in \underline{\mathbf{P}}$
\end{proposition}
\begin{proof}
Assume otherwise, $L^*_P$ ($L^*$ under $P$) is not constricted from above. By lemma \ref{lemma3.270}, there exists $u_k\to 1$, $l_k\in l_P^*(u_k)$ and that $l_k\to +\infty$. By the optimality of $l_k$, $EP(u_k,l_k)\geq EP(u_k,l_A)$. By lemma \ref{lemma3.5},
		\begin{eqnarray}
			\liminf_{k\to +\infty} EP(u_k,l_k)\geq \lim_{k\to +\infty} EP(u_k,l_A)=cp_A(l_A).
		\end{eqnarray}
		Furthermore, associated $I_k=I(u_k,l_k) \to +\infty$. Otherwise, $\exists$ subsequence $I_{k_n}\to \overline{I}<+\infty$. But this implies that $P(I_{k_n})[u_{k_n}p_A(l_{k_n})+(1-u_{k_n})p_B(l_{k_n})]\to 0$, and contradicts optimality of $l_{k_n}$.
		
		Now
		\begin{eqnarray}
			\label{eqn3.24}
			\liminf_{k\to +\infty} EQ(u_k,l_k)\geq a\liminf_{k\to +\infty} EP(u_k,l_k)\geq acp_A(l_A) \mbox{ with } a>1.
		\end{eqnarray}
		On the other hand, $L^*_Q$ is constricted from the above. By lemma \ref{lemma3.5}, $EQ(u_k,l^*_Q(u_k))\to cp_A(l_A)$. This again contradicts inequalities in \ref{eqn3.24}.
		
Similarly we can prove $L^*_P$ must be constricted from below.
\end{proof}

\subsubsection{Computation details}
\begin{lemma}
	We could rewrite that 
	\begin{eqnarray}
		I(u,l)=\frac{(1-u)l\log l}{u+(1-u)l}-\log(u+(1-u)l)
	\end{eqnarray}
	It is easy to compute that
	\begin{eqnarray}
		\label{eqn3.4}
		\frac{\partial I}{\partial l}=\frac{u(1-u)\log l}{[u+(1-u)l]^2}.
	\end{eqnarray}
\end{lemma}

Sometimes we would like to work with $\tilde{l}$ which represents the likelihood ratio of state being $A$ over being $B$. Observing $\tilde{l}$ is equivalent to observing $l=\frac{1}{\tilde{l}}$, thus the generated information is 
\begin{eqnarray}
	\label{eqn3.7}
	I(u,\frac{1}{\tilde{l}})=\frac{u\tilde{l}\log \tilde{l}}{u\tilde{l}+(1-u)}-\log(u\tilde{l}+(1-u))
\end{eqnarray}
which we shall denote as $\tilde{I}(u,\tilde{l})$. 
This computation result has some immediate corollary
\begin{corollary}
	\label{coro3.60}
	\begin{enumerate}
		\item Fixing $u$, $I(u,l)$ strictly decreases on $(0,1)$ and strictly increases on $(1,+\infty)$.
		\item $I(u,l)=0 \Rightarrow l=1.$
		\item $I(u,l)\leq \max\{-\log u,-\log (1-u)\}$.
	\end{enumerate} 
\end{corollary}
\begin{proof}
	Part (1) obviously follows from \ref{eqn3.4}. 
	
	For part (2), as $I(u,1)=0$, if there exists another $\hat{l}$ such that $I(u,\hat{l})=0$, then there is $\tilde{l}$ between $\hat{l}$ and $1$ such that $\frac{\partial I}{\partial l}(u,\tilde{l})=0$. Such a $\tilde{l}$ doesn't exists according to \ref{eqn3.4}.
	
	Part (3) follows from part (1), as $-\log u=\lim_{l\to 0} I(u,l)$. Using \ref{eqn3.7}, $\lim_{l\to +\infty} I(u,l)=\lim_{\tilde{l}\to 0} \tilde{I}(u,\tilde{l})=-\log(1-u)$.
\end{proof}

\begin{lemma}
	\label{lemma3.80}
	Let $u_k\to 1$ and $\{l_k\}\subset (0,M)$ being uniformly bounded from above, then associated 
		\begin{eqnarray}
			I_k=I(u_k,l_k)\to 0.
		\end{eqnarray}
\end{lemma}
\begin{proof}
	Using part (1) in corollary \ref{coro3.60}, 
	\begin{eqnarray}
		I_k\leq \max\{-\log u_k, I(u_k,M)\}.
	\end{eqnarray}
	It is direct to verify that $\lim_{u_k\to 1} \max\{-\log u_k, I(u_k,M)\}=0$.
\end{proof}

\begin{lemma}
	\label{lemma3.270}
	Assume that $L^*$ is not constricted from above under payoff function $Q(I)$. Then if $l_k\in l^*_{Q}(u_k)$ and $l_k\to +\infty$, we have $u_k\to 1$.
\end{lemma}
\begin{proof}
	If $u_k\to 1$ is not true, then there exists subsequence $u_{k_n}$ such that $u_{k_n}\to \overline{u}\in [0,1)$.
	
	In the case that $\overline{u}=0$, using \ref{eqn3.7}, the associated information $I_{k_n}\to 0$. Then $EQ(u_{k_n},l_{k_n})\to 0$. Lemma \ref{lemma3.6} says that a positive payoff of $cp_B(l_B)$ is guaranteed by choose $l_B$. Thus, $l_{k_n}$ cannot be optimal. Contradiction!
	
	In the case that $\overline{u}\in (0,1)$, by part (3) of corollary \ref{coro3.60}, $I_{k_n}$ is bounded away from $+\infty$. This further implies that $P(I_{k_n})$ is bounded away from $+\infty$. Then $EQ(u_{k_n},l_{k_n})\to 0$. Similarly, by fixing the choices to be $1$, 
	\begin{eqnarray}
		EQ(u_{k_n},1)\to cp_A(1)>0.
	\end{eqnarray}
	Again, $l_{k_n}$ cannot be optimal.
\end{proof}

\bibliographystyle{ecta-fullname}
\bibliography{observlearning.bib}
\end{document}